\documentclass[11pt]{article}
\usepackage{amsmath}
\usepackage{amsthm}
\usepackage{amsfonts}
\usepackage{mathpartir}
\usepackage{nuprl}
\usepackage{url}
\newcommand{\skipit}[1]{}
\newcommand\inl[1]{ {\bf inl}\ #1}
\newcommand\inr[1]{ {\bf inr}\ #1}
\newcommand\genprincipal[2]{#1(#2)}
\newcommand\principal[1]{\genprincipal{evd}{#1}}
\newcommand\prin[1]{ {\it principal}(#1)}
\newcommand\newprincipal[1]{evd(#1)}
\newcommand\lam[2]{ \lambda#1.#2}
\newcommand\myall[2]{ \forall#1.\, #2}
\newcommand\uall[3]{ \forall[#1:\!#2].\, #3}
\newcommand\myex[2]{ \exists#1.\, #2}
\newcommand\isect[2]{ \bigcap_{#1} #2}
\newcommand\cbvapp[2]{ (#1)({\bf val}\, #2)}
\newcommand\cbvap[2]{ #1({\bf val}\, #2)}
\newcommand\cbvpair[2]{ \langle{\bf val}\, #1,#2\rangle}

\newcommand\subst[3]{#1[#2:= #3]}
\newcommand\ite[3]{ ({\bf if}\ #1\ {\bf then}\ #2\ {\bf else}\ #3)}
\newcommand\decide[3]{ {\bf decide}(#1;#2;#3)}
\newcommand\spread[2]{ {\bf spread}(#1;#2)}
\newcommand\conts[2]{ \{#1\}\subseteq #2}
\newcommand\mypair[2]{ \langle #1,#2\rangle}
\newtheorem{mydef}{Definition}
\newtheorem{mylemma}{Lemma}
\newtheorem{mythm}{Theorem}
\newtheorem{mycor}{Corollary}
\newtheorem{myexamp}{Example}
\newtheorem{myfact}{Observation}
\newtheorem{myassum}{Assumption}
\newcommand\Lang{\mathcal{L}}
\newcommand\Fml{\mathcal{F}(\Lang)}
\newcommand\Struct{\mathcal{S}(\Lang)}
\newcommand\MFml{\mathcal{MF}(\Lang)}
\newcommand\Fls{\it False}

\newcommand\Var{\it Var}

\newcommand\Void{\it Void}
\newcommand\Term{\it Term}
\newcommand\Unit{\it Unit}

\newcommand\const[1]{ {\bf constant}(#1)}
\newcommand\nat{\mathbb{N}}
\newcommand\stuck{ {\bf abort}}
\newcommand\graph[1]{ {\bf graph}(#1)}
\newcommand\lookup[1]{ {\bf lookup}(#1)}
\newcommand\nsub[1]{\nat_{#1}}
\newcommand\depfun[3]{#1:\!#2 \rightarrow #3}
\newcommand\depprod[3]{#1:\!#2 \times #3}
\begin{document}

\title{
\textbf{Intuitionistic Completeness of First-Order Logic} }

\author{Robert Constable and Mark Bickford}

\maketitle

\begin{abstract}

\noindent We establish completeness for \emph{intuitionistic
first-order logic, iFOL}, showing that a formula is provable if and
only if its embedding into minimal logic, \emph{mFOL}, is
\emph{uniformly valid} under the \emph{Brouwer Heyting Kolmogorov
(BHK)} semantics, the intended semantics of iFOL and mFOL. Our proof
is intuitionistic and provides an effective procedure $P\!r\!f$ that
converts uniform minimal evidence into a formal first-order proof.
We have implemented $P\!r\!f$. Uniform validity is defined using the
intersection operator as a universal quantifier over the domain of
discourse and atomic predicates. Formulas of \emph{iFOL} that are
uniformly valid are also intuitionistically valid, but not
conversely. Our strongest result requires the Fan Theorem; it can
also be proved classically by showing that $P\!r\!f$ terminates
using K{\"o}nig's Theorem.\\

\noindent The fundamental idea behind our completeness theorem is
that a single evidence term \emph{evd} witnesses the uniform
validity of a minimal logic formula \emph{F}. Finding even one
uniform realizer guarantees validity because $P\!r\!f(F,evd)$ builds
a first-order proof of \emph{F}, establishing its uniform
validity and providing a purely logical normalized realizer.\\

\noindent We establish completeness for \emph{iFOL} as follows.
Friedman showed that \emph{iFOL} can be embedded in minimal logic
(\emph{mFOL}) by his A-transformation, mapping formula $F$ to $F^A$.
If $F$ is uniformly valid, then so is $F^A$, and by our completeness
theorem, we can find a proof of $F^A$ in minimal logic. Then we
intuitionistically prove $F$ from $F^{False}$, i.e. by taking
$False$ for $A$ and for $\perp$ of mFOL. Our result resolves an open
question posed by Beth in 1947. \\

\end{abstract}

\section{Introduction}

\subsection{Overview}

\paragraph{\textbf{approaches to completeness}} We introduce a new
approach to completeness questions. It provides the first
intuitionistic completeness proof for the intended semantics of
intuitionistic logic, a question investigated by Beth \cite{Beth47}
starting in 1947 and open ever since.\footnote{See Troelstra
\cite{Tro91} where he states on page 12 ``The standard informal
interpretation of logical operators in intuitionistic logic is the
so-called proof-interpretation or Brouwer-Heyting-Kolmogorov
interpretation (BHK-interpretation for short).'' Brouwer proposed
several interpretations of negation (see \cite{vSt90}), so minimal
logic represents the stable intended core from which it is possible
to explain the ``ex falso quodlibet'' rule as we show.} Our result
provides an answer, however not the one expected by comparison with
G{\"o}del's completeness proof for classical first-order logic. We
briefly review previous completeness results below.

We came to our approach because we use on a daily basis the fact
that from constructive proofs of  a theorem in computational type
theory we can automatically extract programs that meet the
specification given by the theorem. These polymorphically typed
programs are evidence for validity of the theorem. For
intuitionistic first-order logic (iFOL), a subtheory of type theory,
the extracted programs are \emph{uniform} witnesses for validity of
the theorems. We call them \emph{uniform realizers}. We can express
this uniformity by a universal quantifier defined using the
\emph{intersection type} in computational type theory \cite{ABC06}.
Moreover for first-order logic we know that the realizers are not
only uniform, but they are in normal form and consist entirely of
logical operators. This is a basic fact about the extraction of
computational content (see \cite{book-full,ML82}).

In many cases we could see clearly the proof structure in the
realizers. This led us to conjecture that iFOL is complete with
respect to uniform semantics because uniformity eliminates terms
that are not essentially built from the logical operators. It was a
longer road to establish this in detail, and we report succinctly on
that journey here, giving all of the technical details. In a longer
forthcoming article we will provide more motivation, examples, and
practical applications under the proposed title \emph{Intuitionistic
Completeness of First-Order Logic with Respect to Uniform Evidence
Semantics}. There we also explicitly prove some of the basic results
about extraction in the simple setting of iFOL.

The common approach to first-order completeness is based on
systematic search for counter examples to a conjecture, and validity
of the conjecture is the reason the search fails -- halting with a
proof. This approach is well illustrated in Smullyan's enduringly
valued monograph \emph{First-Order Logic} \cite{Smu68} and Fitting's
monograph \cite{Fitt69}, both going back to the work of Beth
\cite{Beth47,Beth57}.\footnote{Beth invented \emph{semantic tableau}
as a bridge from semantics to proofs; we use uniform realizers and
their \emph{evidence structures}.} Like all other classical proofs
of completeness, these are not constructively valid. We take a very
different approach, effectively converting uniform evidence for
validity into a proof. We do this by building objects called
\emph{evidence structures} that reveal the evidence term layer by
layer. For instance, when we see evidence of the form
$\lambda(x.b(x))$ for a formula $A \Rightarrow B$, then we add to
the context of the evidence structure the assumption that $x:A$ and
continue by analyzing $b(x)$ after normalizing it by symbolic
computation. This computation reveals the operations that must be
performed on the context to expose more of the evidence term $b(x)$.
For example, if the assumption $A$ is $A_1 \& A_2$, then we
decompose $x$ into $x_1:A_1$ and $x_2:A_2$ and substitute the pair
$<x_1,x_2>$ into the logical operator mentioning $x$ in the evidence
term we are analyzing. Because the evidence is uniform, the
normalization process eliminates any operators on non-logical terms.
We can thus convert the operators on evidence terms to proof steps
in first-order minimal logic.

Our realizers are \emph{effectively computable} functions operating
on data types; we call this approach \emph{Brouwer realizability} or
\emph{evidence semantics}. We do not rely on Church's Thesis for any
of these results, and according to Kleene \cite{KV65,Tro98b}, our
use of the Fan Theorem precludes it.\footnote{The Computational Type
Theory which Nuprl implements was designed in 1984 to use an
\emph{open-ended} notion of effective computability from the start
\cite{book-full}.}

\skipit{Our method of establishing completeness with respect to
evidence semantics is useful even for logics that are incomplete,
such as Heyting Arithmetic, because our proof building procedure
will work in a practical way on uniform evidence in theories known
to be incomplete. \emph{Our method converts uniform semantic
evidence for a proposition into a proof.} Completeness by our method
will hold for any class of theorems and evidence terms for which our
basic procedure and its variants halts.}

We hope that our results will add more weight to the notion that
there is a deep connection between proving a theorem and writing a
program. We have long stressed this idea in papers treating
\emph{proofs as programs} \cite{ABC06,BC85, book-full} and
conversely \emph{programs as proofs}, additionally in papers
treating formal constructive \emph{mathematics as a programming
language} \cite{Con71, book-full} where types subsume data types.
Here we are treating iFOL as an abstract programming language where
formulas are specifications given by dependent types. We build the
proof from the program/data type which is a uniform Brouwer
realizer.

\paragraph{\textbf{intuitionistic model theory}} This article contributes
to an \emph{intuitionistic model theory} as proposed by Beth in 1947
\cite{Beth47} and greatly advanced by Per Martin-L{\"o}f \cite{ML82,
ML84}. Beth's methods led to \emph{Beth models} and \emph{Kripke
models} whose computational meaning is not as strong as in the
realizability tradition, even given Veldman's intuitionistic
completeness theorem for Kripke models \cite{Veld76}. We work in the
realizability tradition started by Kleene, developed further by
Martin-L{\"o}f, extended and implemented by the PRL Group as
reported in the book \emph{Implementing Mathematics}
\cite{book-full}, by the Coq Group as reported in \cite{BYC04}, the
Gothenberg Group reported in the book \emph{Programming in
Martin-L{\"{o}}f's Type Theory} \cite{NPS90}, the Minlog Group as
reported in Proof Theory at Work: Program Development in the
{Minlog} System \cite{BBS98}, and in numerous doctoral dissertations
and articles many of which are cited in \cite{ABC06}. This is the
tradition framing and motivating our completeness results.

The semantic tradition is grounded in precise knowledge of the
underlying computation system and its efficient implementation made
rigorous by researchers in programming languages. Our operational
semantics of evidence terms follows the method of \emph{structured
operational semantics} of Plotkin \cite{Plo77,Plo81}. The few basic
results about programming language semantics we mention can be found
in the comprehensive textbooks on the subject \cite{Mit96,Pie02}.
Many results from this theory are now being formalized in proof
assistants and applied directly to building better languages and
systems \cite{Pie11}.

\skipit{Our results can be applied in this enterprize as well and
are related to research using the Agda programming language and
prover \cite{BDN09}. Moreover, this work on semantics will make it
feasible to completely formalize the proof given here in strong
theories such as CTT and ITT augmented with a realizer for the Fan
Theorem.}

\subsection{Background}

\paragraph{Classical first-order logic, FOL} Tarski's semantics \cite{Tar38} for classical
first-order formulas faithfully captures their intuitive
\emph{truth-functional interpretation}. G{\"o}del proved his
classical completeness theorem for first-order logic with respect to
this intended semantics, showing that an FOL formula is provable if
and only if truth functionally valid. This has become a fundamental
result in logic which is widely taught to undergraduates. There are
many excellent textbook proofs such as \cite{Smu68}.

\paragraph{Intuitionistic first-order logic, iFOL}

The BHK semantics for iFOL is the intended semantics, faithful to
the intuitionistic conception of knowledge. In contrast to the
classical situation, there has been no intuitionistic completeness
proof with respect to the intended semantics.  To explain this
contrast, we look briefly at the origin of intuitionism. At nearly
the same time that a truth-functional approach to logic was being
developed by Frege \cite{Fre67} and Russell \cite{Rus08},
\emph{circa} 1907, Brouwer \cite{Bro75,vSt90} imagined a very
different meaning for mathematical statements and thus for logic
itself. Brouwer's meaning is grounded in the \emph{mental
constructions} that cause an \emph{individual mathematician} to
\emph{know} that mathematical objects can be created with certain
properties.

Brouwer developed a very rich informal model of computation in terms
of which he could interpret most concepts and theorems of
mathematics, including from set theory (see \cite{vSt90}). Brouwer's
approach anticipated a precise meaning that Church, Turing, and now
legions of computer scientists give to mathematical statements whose
meaning is grounded in computations executed by modern digital
computers. Brouwer's intuitive interpretation has come to be known
among logicians as \emph{Brouwer, Heyting, Kolmogorov (BHK)
semantics} when applied to formal intuitionistic logical calculi, as
first done by Heyting \cite{Hey34} and Kolmogorov \cite{Kol32}. In
1945 Kleene \cite{Kle45,Tro98b} invented his \emph{realizability}
semantics for intuitionistic number theory in order to connect
Brouwer's informal notion of computability to the precise theory of
partial recursive functions. He used indexes of general recursive
functions as \emph{realizers}, and by 1952 \cite{Kle52} he viewed
realizability as a formal account of BHK semantics under the
assumption of Church's Thesis.

By 1982 Martin-L{\"o}f \cite{ML82,ML84} building on the work of
Kleene refined the BHK approach and raised it to the level of a
\emph{semantic method} for constructive logics grounded in
structured operational semantics \cite{Plo81}. Martin-L{\"o}f often
referres to BHK as the \emph{propositions as types principle}. In
computer science other terminology is ``proofs as programs'' or the
``Curry-Howard isomorphism''. Already in 1970 Martin-L{\"o}f
proposed using Brouwer's analysis of bar induction as the meaning of
$\Pi^1_1$ statements and developed a constructive version of
completeness for classical first-order logic \cite{ML70} based on a
topological model of Borel sets in the Cantor Space.

This semantics plays an important role in the business of building
\emph{correct by construction} software and in the semantics of the
constructive logics such as \emph{Computational Type Theory} (CTT)
\cite{ABC06, book-full}, \emph{Intuitionistic Type Theory} (ITT)
\cite{ML82,ML84,NPS90}, \emph{Intensional-ITT}
\cite{BDN09,ML98,NPS90}, the \emph{Calculus of Inductive
Constructions} (CIC)\cite{BYC04}, \emph{Minlog} \cite{BBS98}, and
\emph{Logical Frameworks} such as Edinburgh LF \cite{HHP93}. All of
these logics are implemented by \emph{proof assistants} such as
Agda, Coq, MetaPRL, Minlog, Nuprl, and Twelfth among others.

\paragraph{Previous completeness theorems} A constructive
completeness theorem for iFOL with respect to \emph{intuitionistic
validity} is a very strong result because it says that if we know
that a formula is valid, thus true in every possible model,
\emph{then we can effectively find a first-order proof based on that
knowledge}. This seems highly unlikely as the sixty four year long
investigation of the problem has shown. In all previous work, the
idea is to try to construct a proof and use the evidence for truth
to argue that the proof construction must succeed.  Classically this
requires K{\"o}nig's Lemma, and constructively some use of Markov's
Principle or the Fan Theorem or something of that kind. Those
efforts do not try to use the information that $\forall
\mathcal{M}:Model. \models F$ to build the proof.  Nevertheless, our
results show exactly how to build the proof from \emph{uniform
evidence} for validity, which is a single object. Moreover, we can
actually execute our result using a tactic executed by the Nuprl
prover \cite{book-full, ABC06, Kre02}. We give that procedure in the
Appendix.

Over the last fifty years there have been numerous deep and
evocative efforts to formulate completeness theorems for the
intuitionistic propositional calculus and for intuitionistic
first-order logic modeled after G{\"o}del's Theorem
\cite{deS77,Dumm77,Kre62,ML70,Veld76,McC08}.  Some efforts led to
apparently more technically tractable semantic alternatives to BHK
such as \emph{Beth models} \cite{Beth57,Veld76}, \emph{Kripke
models} \cite{Kri65}, topological models
\cite{CoqSmi95,Grz64b,Tar38,RS63,ML70}, intuitionistic model
theoretic validity \cite{TvD88}, and provability logic
\cite{Art99b}. Dummett \cite{Dumm77} discusses completeness issues
extensively. The value of developing a precise mathematical
semantics for intuitionistic mathematics in the spirit of Tarski's
work dates at least from Beth 1947 \cite{Beth47} with technical
progress by 1957 \cite{Beth57}. So the completeness issue has been
identified yet unsettled for sixty four years. An important early
attempt to base completeness on BKH is the (nonconstructive) work of
L{\"a}uchli \cite{Lau70, LO96} who stressed the notion of
\emph{uniformity} as important. None of these efforts provides a
constructive completeness theorem faithful to BHK semantics (a.k.a.
Brouwer realizability) either for the intuitionistic propositional
calculus (IPC) or for the full predicate calculus. We do.

The closest correspondingly faithful constructive completeness
theorem for \emph{intuitionistic validity} is by Friedman in 1975
(presented in \cite{TvD88}), and the closest classical proof for the
Brouwer-Heyting-Kolmogorov (propositions as types/proofs as
terms/proofs as programs) semantics for intuitionistic first-order
logic is from 1998 by Artemov using \emph{provability logic}
\cite{Art99b}. Results suggest how delicate completeness theorems
are since constructive completeness with respect to full
intuitionistic validity contradicts Church's Thesis
\cite{Kre62,TvD88} and implies \emph{Markov's Principle} as well
\cite{McC96,McC08}.\footnote{Church's Thesis is not an issue for us
because we do not assume it.}

\subsection{Summary of Results}

\paragraph{Results in this article} We first review \emph{evidence semantics}.
\footnote{We can extend this semantics to classical logic if we
allow oracle computations \cite{Con85e} to justify the law of
excluded middle, $P \vee \sim \!P$, with an operator $magic(P).$ We
make some observations about classical logic based on this
\emph{classical evidence semantics}.} Using evidence semantics, we
then introduce the idea of \emph{uniform validity}, a concept
central to our results and one that is also classically meaningful.
This concept provides an effective tool for semantics because we can
establish uniform validity by \emph{exhibiting a single polymorphic
object}. For example, the propositional formula $A \Rightarrow A$ is
uniformly valid exactly when there is an object in the intersection
of all evidence types for this formula for each possible choice of
$A$ among the type of propositions, $\mathbb{P}$. We write this
intersection as $\forall [A:\mathbb{P}]. A \Rightarrow A$ or as
$\bigcap A:\mathbb{P}. A \Rightarrow A$.\footnote{We work in a
\emph{predicative} metatheory, therefore the type of all
propositions is stratified into orders or levels, written
$\mathbb{P}_i$. For these results we can ignore the level of the
type or just write $\mathbb{P}_i$.} In this case, given the
extensional equality of functions, the polymorphic identity function
$\lambda(x.x)$ is the one and only object in the intersection. So
the witness for uniform validity like the witness for provability,
can be provided by a single object.\footnote{Contrast this with the
kind of evidence need for classical or intuitionistic \emph{model
theoretic validity}. In those cases, we need a whole class of models
to witness validity of a single formula.} Truth tables do this for
classical propositional logic. Unlike for classical first-order
logic, there are single witnesses for the validity of all uniformly
valid first-order formulas; for example, it will be clear after we
provide the evidence semantics that the polymorphic term
$\lambda(h.\lambda(x.\lambda(p.h(<x,p>))))$ establishes the uniform
minimal (logic) validity of

$$\sim \exists x .P(x) \Rightarrow \forall x .(\sim \! P(x))$$
hence the uniform intuitionistic and classical validity as well.

Another important observation about uniform validity is that
\emph{the formulas of first-order logic that are provable
intuitionistically and minimally are uniformly valid.} It is also
noteworthy that \emph{the law of excluded middle is not uniformly
valid in either constructive or classical evidence semantics.}

Uniform validity also raises the semantic problem that forces us to
consider minimal logic first. Consider the intuitionistically valid
assertion $False \Rightarrow A$ for any proposition $A$. One
semantic object that witnesses uniform validity is $\lambda(x.x)$,
and other witnesses for uniform validity include any constant
functions, say $\lambda(x.17)$ or even a \emph{diverging} term such
as $div$. The claim being made is that if $x$ belongs to the
evidence type for $False$, then 17 or $div$ belongs to the evidence
type for $A$.\footnote{We can use the fixed point combinator, say
\textbf{Y} or $fix$ to define $div$. For instance,
$fix(\lambda(x.x))$ computes to itself, where $fix$ is an operator
such as the \textbf{Y} combinator
$\lambda(f.ap(\lambda(x.ap(f;ap(x;x)));\lambda(x.ap(f;ap(x;x)))))$.}
This claim is vacuously true since no element can be evidence for
$[False]$ whose evidence is the empty type. From the constant
function with value 17, $\lambda(x.17)$, we cannot reconstruct the
proof. In minimal logic, we don't have the atomic propositional
constant $False$, we use instead the \emph{arbitrary} propositional
constant \textbf{$\perp$} whose interpretation allows non-empty
types as well as empty ones. For the same reason, avoiding vacuous
hypotheses, we require that all domains of discourse for minimal
logic can be non-empty.

\paragraph{Discussion} Our results also suggest why completeness with respect to
satisfiability in all constructive models, let alone all
intuitionistic models, is unlikely (even impossible according to
McCarty \cite{McC96,McC08}); such completeness is unlikely because
we show that provability captures exactly \emph{uniform validity},
an intuitively smaller collection of formulas than those
constructively valid. Nevertheless, uniform validity is extremely
useful in practice when thinking about purely logical formulas
precisely because it corresponds exactly to proof and yet is an
entirely semantic notion based on the intended BHK semantics, the
semantics that enables strong connections to computer science.

\section{The main theorems}
\begin{mydef}
  A first order language $\Lang$ is a symbol $D$ and a finite set of relation symbols $\left\{ R_i \vert i \in I \right\}$
  with given arities $\left\{ n_i \vert i \in I \right\}$.
  First order formulas, $\Fml$, over $\Lang$ are defined as usual.
  The variables in a formula (which range over $D$) are taken from a fixed set $\Var = \left\{ d_i \vert i \in \mathbb{N} \right\}$.
  Negation $\neg \psi$ can be defined to be $\psi \Rightarrow \Fls$.
  The first order formulas of minimal logic
  \footnote{The usual definition of minimal logic includes a designated constant $\perp$ and defines weak negation as $\psi \Rightarrow \perp$.
            We merely view $\perp$ as one of the atomic relation symbols $R_i$ with arity $n_i = 0$.
           },
  $\MFml$, are the formulas in $\Fml$ that do not use either negation or $\Fls$.
\end{mydef}

  In type theory, the propositions, $\mathbb{P}$, are identified with types. A non-empty type is a true proposition and members of
  the type are the evidence for the truth of the proposition. An empty type provides no evidence and represents a false proposition.

\begin{mydef}
  A structure $M$ for $\Lang$ is a mapping that assigns to $D$ a type $M(D)$ and assigns to each $R_i$ a term of type
  $M(D)^{n_i} \rightarrow \mathbb{P}$.
  We write $\Struct$ for the type\footnote{Since we work in type theory we always use types rather than sets.} of structures for $\Lang$.
  If $M \in \Struct$ and $x \in M(D)$ then $\subst{M}{d}{x}$ is an extended structure that maps the variable $d$ to the term $x$.
\end{mydef}

\begin{mydef}
  \label{model1}
  Given $M \in \Struct$ that has been extended to map the variables $V_0 \subseteq \Var$
  into $M(D)$, we extend the mapping $M$ to all formulas in $\Fml$ with free variables in $V_0$ by:
  \begin{eqnarray*}
    M(\Fls) &=& \Void\\
    M(R_i(v_1,\dots,v_{n_i})) &=& M(R_i)(M(v_1),\dots, M(v_{n_i})) \\
    M(\psi_1 \wedge \psi_2) &=&  M(\psi_1) \times M(\psi_2) \\
    M(\psi_1 \vee \psi_2) &=&  M(\psi_1) + M(\psi_2) \\
    M(\psi_1 \Rightarrow \psi_2) &=&  M(\psi_1) \rightarrow M(\psi_2) \\
    M(\neg\psi) &=&  M(\psi \Rightarrow \Fls) \\
    M(\myall{v}{\psi}) &=&  \depfun{x}{D}{(\subst{M}{v}{x})(\psi)} \\
    M(\myex{v}{\psi}) &=&  \depprod{x}{D}{(\subst{M}{v}{x})(\psi) }
  \end{eqnarray*}
  Thus, any $M \in \Struct$ assigns a type $M(\psi)$ to a
  sentence (a formula with no free variables) $\psi \in \Fml$. $M(\psi)$ is
  synonymous with the proposition $M \models \psi$, and the members of type $M(\psi)$ are the
  evidence for $M \models \psi$.
\end{mydef}
\begin{mydef}
  A sentence $\psi \in \Fml$ is valid if \[\myall{M \in \Struct}{M \models \psi}\]
  Evidence for the validity of $\psi$ is a function of type $M:\Struct \rightarrow M(\psi)$
  that computes, for each $M \in \Struct$, evidence for $M \models \psi$.

  A sentence $\psi \in \Fml$ is uniformly valid if there is one term that is a member
  of all the types $M(\psi)$ for $M \in \Struct$. Such a term is a member of the
  intersection type \[\isect{M\in\Struct}{M(\psi)}\] We write an intersection type $\isect{x\in T}{P(x)}$ as
  a proposition using the notation $\uall{x}{T}{P(x)}$. The square brackets indicate that evidence for the
  proposition  $\uall{x}{T}{P(x)}$ is uniform and does not depend on the choice of $x$.

  To summarize:
  \begin{eqnarray*}
    \psi\ {\it is\ valid} &\equiv& \myall{M \in \Struct}{M \models \psi}\\
    \psi\ {\it is\ uniformly\ valid}(\psi) &\equiv& \uall{M}{\Struct}{M \models \psi}
  \end{eqnarray*}
\end{mydef}
We write $\vdash_{IL} \psi$ to say that there is a proof of $\psi$
in intuitionistic logic and $\vdash_{ML} \psi$ to say that there is
a proof of $\psi$ in minimal logic. From a proof in intuitionistic
logic of any proposition we can construct evidence for the
proposition. Automated proof assistants like
Agda,Coq,MetaPRL,Minlog, and Nuprl can construct the evidence
automatically. We observe, and can easily prove, that the evidence
constructed from an intuitionistic proof of a first order formula
$\psi \in \Fml$ is actually evidence that $\psi$ is uniformly valid.
Our main theorem states that for formulas of minimal logic the
converse is also true: a uniformly valid formula is provable.
\begin{mythm}
  \label{thm:min}
  For any $\psi \in \MFml$,
  \[ \uall{M}{\Struct}{M \models \psi} \Leftrightarrow\ \  \vdash_{ML}
  \psi . \]
\end{mythm}
Using Friedman's $A$-transformation \cite{Lei85}, we can derive from
Theorem~\ref{thm:main} a corresponding completeness theorem for
intuitionistic logic.
\begin{mycor}
  \label{thm:main}
  For any $\psi \in \Fml$,
  \[ \uall{M}{\Struct}{M \models \psi^A} \Leftrightarrow\ \
\vdash_{IL} \psi \]
\end{mycor}
\begin{proof}
  By Theorem~\ref{thm:min} it is enough to show
  \[ \vdash_{ML} \psi^A \Leftrightarrow\ \  \vdash_{IL} \psi \]
  ($\Rightarrow$) If $\vdash_{ML} \psi^A$ then also $\vdash_{IL}
\psi^A$ for any interpretation of $A$ including ${\bf False}$. It is
easy to
  prove, by induction on the structure of $\psi$ that $\vdash_{IL}
\psi^{\bf False} \Leftrightarrow \psi$.

  ($\Leftarrow$) This is Friedman's Theorem.
\end{proof}
We will prove Theorem~\ref{thm:min} by defining
  an effective procedure that builds a tree of {\em evidence structures} (defined below) starting with an initial
  evidence structure formed from the uniform evidence term. We show that any evidence structure is either trivial (and therefore
  a leaf of the ultimate minimal logic proof) or else can be transformed into a finite number (either one or two) of derived
  evidence structures, and the transformation tells us what rule of minimal logic to use at that step of the proof.

Theorem~\ref{thm:min} will then follow from the fact that this
effective procedure must terminate and yield a finite proof tree.
The termination of the procedure for an arbitrary term $evd \in
\isect{M\in\Struct}{M(\psi)}$ is a strong claim. The evidence need
not be a fully-typed term with all of its subterms typed, so there
can be sections of ``dead code'' in the evidence that are not
typable and may not be normalizable. Nevertheless the fact that the
evidence is uniformly in the type $M(\psi)$ implies that the ``dead
code'' is irrelevant and our procedure will terminate, but the proof
of this fact (which follows in classical logic from K{\"o}nig's
lemma) in intuitionistic mathematics seems to require Brouwer's Fan
Theorem.

If we assume that the uniform evidence term is fully normalized,
then we can make a direct inductive argument for termination of our
proof procedure. Since the evidence constructed from a proof in
minimal logic is fully normalizable, this results in an alternate
version of completeness that we state as follows
\begin{mythm}
  \label{thm:normal}
  Any $\psi \in \MFml$ is provable in minimal logic ($\vdash_{ML} \psi$) if and only if there is a fully normalized term
  $evd$ in the type $\isect{M\in\Struct}{M(\psi)}.$
\end{mythm}
We work only in intuitionistic logic, so we must avoid the use of
excluded middle for propositions that are not decidable and in
particular we can not assume the proposition $evd \in M(\psi)$ is
decidable. Because of this, we will need the concept that evidence
term $evd$ is {\em consistent with} the type $M(G)$. One notion of
consistency that is sufficent for our proof is that there is no
structure $M$ for which $evd \not\in M(G)$. However, the resulting
proof is logically more complex than the one we give below where
consistency is based on interpeting the types in {\em finitary}
structures.

\section{Finitary types and structures}
\begin{mydef}
  Types $A$ and $B$ are {\em equipollent} (written $A \sim B$) if there is a bijection $f: A \rightarrow B$.
  A type $T$ is {\em finite} if $\myex{k:\nat}{T \sim \nsub{k}}$ (where $\nsub{k}$ is the
  type of numbers in the range $0 \leq i < k$).
\end{mydef}
Note that if $T$ is finite then equality in $T$ is decidable and there is a list $L_T$ that enumerates $T$, i.e. contains all
the members of $T$ with no repeats. Using $L_T$, any function $f : T \rightarrow S$ can be converted to a table
\[\graph{f} = {\bf map}(\lambda x. \mypair{x}{f(x)}, L_T) \]
Using the decidable equality in $T$ we can define a table lookup function and recover the
function \[f = \lookup{\graph{f}}\]
\begin{mydef}
  We write $t\downarrow$ to say that term $t$ computes to a value (a canonical form).

  The {\em bar type} $\bar{T}$ is the type of all terms $t$  such that $(t\downarrow) \Rightarrow (t \in T)$.

  A function $f : \Term \rightarrow \bar{T}$ is {\em strict} if for all terms $t$
  \[(f(t)\downarrow) \Rightarrow (t\downarrow) \]

  A type $T$ is a {\em value type} if every member of $T$ converges to a value.
\end{mydef}
A bar type $\bar{T}$ is not a value type, but even without bar types, a rich type theory that includes intersection types or
quotient types will have some types that are not value types.

\begin{mydef}
  A type $T$ is a {\em retract} if there is a strict function $i_T$ of type $\Term \rightarrow \bar{T}$ such that
  \[\myall{t:T}{i_T(t) = t \in T}\] or equivalently \[i_T = id \in (T \rightarrow T)\]

  A type $T$ is {\em finitary} if it is a finite, value type and a retract.

  A structure $M \in \Struct$ is {\em finitary} if $M(D)$ is finitary and the types
    $M(R_i)(d_1,\dots, d_{n_i})$ assigned to the atomic formulas are finitary.
\end{mydef}
We let $\stuck$ be a fixed term that has no redex but is not a
canonical form. For example $\stuck$ could be $0(0)$ or $true + 5$
or a primitive. The term $\stuck$ does not converge to a value. We
use this to construct simple examples of finitary types.
\begin{myexamp}
The type $\nsub{k}$ is a finitary type. The retraction $i_{\nsub{k}}$ is
\[\lambda t.  \ite {0 \leq t\ \&\ t < k}{t}{\stuck}. \]
\end{myexamp}
\begin{myexamp}
  The type $\Unit$ with a single canonical member $\star$ is a finitary type. The retraction is
  \[\lambda t.  \ite {t == \star}{t}{\stuck}. \]
\end{myexamp}
These examples depend on the existence of primitive computations
that recognize the canonical forms of the intended members of the
type. We mention here some additional assumptions about the
underlying computation system on which our proof of completeness
depends. These assumptions are satisfied by the computation system
used by Nuprl, but our proof could easily be modified to work for
type theories based on different primitive computations.
\footnote{For example, the computation system could have primitive
projection functions $\pi_1$ and $\pi_2$ rather than the ${\bf
spread}$ primitive. It could have primitives for ${\bf isl}$,  ${\bf
outl}$, ${\bf isr}$, and ${\bf outr}$ rather than the ${\bf decide}$
primitive. Our construction would be easily modified to accomodate
such differences.}
\begin{myassum}
  The only primitive redex involving a pair $\mypair{t_1}{t_2}$ is
  \[\spread{\mypair{t_1}{t_2}}{x,y.B(x,y)} \mapsto B(t_1,t_2)\]
  The only primitive redex involving $\inl{t}$ is
  \[\decide{\inl{t}}{x.B(x)}{y.C(y)} \mapsto B(t)\]
  The only primitive redex involving $\inr{t}$ is
  \[\decide{\inr{t}}{x.B(x)}{y.C(y)} \mapsto C(t)\]
  The only primitive redex involving $\lambda x. B(x)$ is
  \[(\lambda x. B(x)) (t) \mapsto B(t)\]
\end{myassum}

\begin{mylemma}
  If $A$ and $B$ are finitary then $A + B$ is finitary.
  If $A$ is finitary and for all $a \in A$, $B(a)$ is finitary then the types
  $\depfun{a}{A}{B(a)}$ and $\depprod{a}{A}{B(a)}$ are both finitary.
  \label{lemma:finitary}
\end{mylemma}
\begin{proof}
  It is straightforward to prove that the types are finite, value types.
  The retraction maps $i$, $j$, and $k$ for $A + B$ ,
  $\depfun{a}{A}{B(a)}$, and $\depprod{a}{A}{B(a)}$ are
  \begin{eqnarray*}
  i(t) &=&  \decide{t}{x.i_A(x)}{y.i_B(y)}\\
  j(t) &=&  \lookup{\graph{\lambda a.i_{B(a)}(t(a))}}\\
  k(t) &=& \spread{t}{x,y.\cbvapp{\lambda a.\cbvapp{\lambda b.\mypair{a}{b}}{i_{B(a)}(y)}}{i_A(x)}}
  \end{eqnarray*}
  The operation $\cbvap{f}{x}$ is a {\em call-by-value apply}, so for the retraction $k(t)$ to converge,
  the term $t$ must evaluate to a pair $\mypair{x}{y}$, the value $a = i_A(x)$ must converge, and the value
  $b = i_{B(a)}(y)$ must converge, before the pair $\mypair{a}{b} \in \depprod{a}{A}{B(a)}$ is formed.

\end{proof}
\begin{mycor}
  If $M \in \Struct$ is finitary and $\psi \in \Fml$ then $M(\psi)$ is finitary.
\end{mycor}
\begin{mydef}
  We abbreviate $\cbvapp{\lambda a.\mypair{a}{y}}{x}$ as $\cbvpair{x}{y}$. This operation forms a pair only after the first component has been evaluated.
\end{mydef}
\begin{mydef}
  A term $t$ is {\em consistent with} a retract type $T$ if $i_T(t) \in T$ or, equivalently, if $i_T(t) \downarrow.$

  If $A$ is a retract then
  a function $f :\! A \rightarrow B$ is {\em tight} if the domain of $f$ contains only terms consistent with $A$, i.e. if
  for all terms $t$ \[(f(t)\downarrow) \Rightarrow (i_A(t)\downarrow).\]
\end{mydef}
\begin{mylemma}
 If $f$ has type $A \rightarrow B$ and $A$ is a retract, then there is a tight function
 $f' = f \in (A \rightarrow B).$
  \label{lemma:consistent}
\end{mylemma}
\begin{proof}
  Let $f' = f \circ i_A $ where $i_A$ is the retraction onto $A$. Since $i_A$ is the identity on $A$, we have
 $f' = f \in (A \rightarrow B)$. The domain of $f'$ contains only terms in the domain of $i_A$.
\end{proof}

\begin{mydef}
  For any $\sigma : V_0 \rightarrow T_0$ that is an injection from a finite subset $V_0 \subset \Var$ into a finitary type $T_0$
  we define a finitary $\Lang$-structure $M_{triv}(\sigma)$ by
  \begin{eqnarray*}
    M(D) &=& T_0\\
    M(v) &=& \sigma(v) \\
    M(R_i) &=& \lambda x_1,\dots,x_{n_i}. \Unit.
  \end{eqnarray*}
\end{mydef}
\begin{mylemma}
  For any $\psi \in \MFml$ with free variables in $V_0$, $M_{triv} \models \psi.$
  \label{lemma:sat}
\end{mylemma}
\begin{proof}
  The structure $M_{triv}(\sigma)$ assigns to every atomic formula the non-empty type $\Unit$.
  It is then clear that $M_{triv}(\sigma)$  assigns a non-empty type to every minimal logic formula $\psi$
  and hence $M_{triv} \models \psi$. This would not be true for general first-order
  formulas that include negation and False.
\end{proof}
\section {Evidence structures}
We will use the concept of an {\em evidence structure} to build a bridge between uniform evidence terms and proofs.
An evidence structure will have three parts, a context $H$, a goal $G$, and evidence term $evd$.
The context $H$ will include some declarations of the form $d_i:D$ where $d_i \in \Var$ (the variables in $\Fml$),
but it will also include declarations of the form $v_i: A$ where $A \in \MFml$ is a subformula of the orginal goal $\psi$
and $v_i$ is a variable chosen from another set of variables $\Var' = v_0, v_1, v_2, \dots$ disjoint from $\Var = \left\{ d_0, d_1, \dots \right\}$.
The context $H$ will also contain {\em constraints} of the form $\cbvap{f}{d} = t$ where $f\in \Var'$ and term $t$ is a {\em pattern} over $H$.
\begin{mydef}
  Given a set $H$ of variable declarations $v:\! T$, the set of
  patterns over $H$ is the set of typed terms defined inductively by:
  \begin{enumerate}
    \item Any $v:\! T \in H$ is a pattern.
    \item If $ptn_1:\!A$ and $ptn_2:\! B$ are patterns then the following are patterns:
      \begin{itemize}
    \item $\mypair{ptn_1}{ptn_2}:(A \times B)$
    \item $\inl{ptn_1}:\!(A + B)$
    \item $\inr{ptn_2}:\!(A + B).$
      \end{itemize}
  \end{enumerate}
\end{mydef}
\begin{mydef}
  A typing $H$ over $\Lang$ is a list of declarations of one of the two forms:
      \begin{enumerate}
    \item $d:\! D$ where $d \in \Var$.
    \item $v:\! A$ where $v \in \Var'$ and $A \in \MFml$ such that every
      free variable $d$ of $A$, is declared in $H$.
      \end{enumerate}
      A {\em model} $M$ of $H$ is a finitary structure for $\Lang$ extended so that for each $v:\!T$ in $H$,
      $M(v) \in M(T)$.
\end{mydef}

\begin{mydef}
  An {\em implies constraint} on a typing $H$ is an equation
  \[v_i = \const{t}\]\
  where $v_i:\! A \Rightarrow B \in H$
  and $t$ is a pattern of type $B$.
  The constraint is {\em stratified} if for any variable $v_j$ in pattern $t$, $i < j$.
  The constraint is {\em unique} in $H$ if there is no other constraint $v_i = \const{t'}$ in $H$.
  A model $M$ of $H$ satisfies the  constraint if $M(v_i) = \lambda x. M(t) \in (M(A)\rightarrow M(B))$.

  A {\em forall constraint} is an equation
  \[\cbvap{v_i}{d} = t\]
  where $d:\!D \in H$ and for some formula $P\in \MFml$, $v_i:\!\myall{z}{P} \in H$ and
  $t$ is a pattern of type $P(d)$ over $H$.
  The constraint is {\em stratified} if for any variable $v_j$ in pattern $t$, $i < j$.
  The constraint is {\em unique} in $H$ if there is no other constraint $\cbvap{v_i}{d} = t'$ in $H$.
  A model $M$ of $H$ satisfies the constraint if $M(v_i)(M(d)) = M(t) \in M(P(d))$.

  An evidence context $H$ over $\Lang$ is a list of declarations and unique, stratified constraints such that
  the declarations are a typing over $\Lang$ and the constraints are constraints on that typing.
  $M$ is a {\em model of context} $H$ if it is a model of the typing $H$ that satisfies all the constraints.
      We write $M \models H$ to say that $M$ is a model of context $H$; note that this means that $M \in \Struct$ and
      $M$ is finitary.
\end{mydef}
\begin{mydef}
  A model $M \models H$ is {\em tight} if for every $f:\! A\rightarrow B \in H$, the function $M(f)$ is tight.
  We write $M \models_t H$ when $M$ is tight.
\end{mydef}
\begin{mylemma}
  Every evidence context $H$ over $\Lang$ has a tight model.
  \label{satisfiable}
\end{mylemma}
\begin{proof}
  \newcommand\Mtriv{M_{triv}(\sigma)}
  Let $V_0$ be the variables $d_i$ for which $d_i:\!D \in H$.
  We first choose a finitary type $T_0$ and an injection $\sigma: V_0 \rightarrow T_0$ (we can use $\nsub{k}$ for $k > \vert V_0\vert$).
  We construct the model $M$ by extending the model $\Mtriv$,
  choosing values for the variables that satisfy all the constraints.
  Since the constraints are stratified, we choose values for the variables in reverse order.
  Let $v_j \in \Var'$ be a variable with a declaration $v_j:T \in H$ and assume that we have chosen values for all variables $v_k$ in $H$
  with $j < k$.
  Assign a value to all patterns $t$ all of whose variables $v_k$ have $k > j$ recursively as follows:
  $M(\mypair{p_1}{p_2}) = \mypair{M(p_1)}{M(p_2)}$, $M(\inl{p}) = \inl{M(p)}$, $M(\inr{p}) = \inr{M(p)}$.

  If $T$ is $\myall{x}{P}$ for some $P$, then for each $d_i \in V_0$ we choose a value $w_i \in M(P(d_i))$ as follows:
  If there is a (unique) constraint $\cbvap{v_j}{d_i} = t_i$ in $H$ then we use $w_i = M(t_i)\in M(P(d_i))$ (which is defined since values for the
  variables in pattern $t_i$ have already been chosen). Otherwise we choose $w_i$ from the non-empty type $M(P(d_i))$.
  Since the values $M(d_i)$ are all distinct members of the finite type $M(D) = T_0$,
  we set the value of $v_j$ to be a function of type $x:M(D) \rightarrow M(P(x))$ that
  maps each $M(d_i)$ to $w_i$. This function is $\lookup{\left[ \mypair{M(d_i)}{w_i} \vert d_i \in V_0\right]}$.

  If $T$ is $A \Rightarrow B$ then if there is a (unique) implies constraint $v_j = \const{t}$ then let $w = M(t)$ and
  choose the constant function $\lambda x.w$ made tight by applying lemma~\ref{lemma:consistent}.
  If there is no constraint on $v_j$ then we choose any member of the non-empty type $M(A) \rightarrow M(B)$
  and make the chosen function tight by applying lemma~\ref{lemma:consistent}.

  Otherwise there are no constraints on $v_j$, and $\Mtriv(T)$ is non-empty
  by lemma~\ref{lemma:sat} so we may choose a value for $v_j$ from this type.
\end{proof}
\begin{mydef}
  An evidence structure is a triple $H \models G, evd$ where
  \begin{enumerate}
     \item $H$ is an evidence context.
     \item $G \in \MFml$ .
     \item for every $M \in \Struct$, if $M \models_t H$ then $M(evd)$ is consistent with $M(G)$.
  \end{enumerate}
We write $\subst{t}{v}{e}$ for the result of substitution of $e$ for variable $v$ in term $t$,
and we write $\subst{(H \models G, evd)}{v}{e}$ for the result of substitution of $e$ for $v$ everywhere in the evidence structure $H\models G, evd$.
\end{mydef}
\begin{myfact}
  If $evd$ is uniform evidence for a formula $\psi \in \MFml$ then \[ \models \psi, evd\] is an evidence structure.
\end{myfact}
\section{Derivation rules for evidence structures}
We now define a set of sixteen derivation rules by which
we derive evidence structures from evidence structures. We will prove that
\begin{enumerate}
  \item If $H \models G, evd$ is an evidence structure, then $evd$ computes to $evd'$ that is canonical or has a principal argument that is a variable.
  \item $H \models G, evd'$ is an evidence structure that it matches one of the sixteen derivation rules.
  \item This defines a recursive procedure on evidence structures that results in a tree of \emph{derived evidence structures}.
  \item The tree derived from $( \models \psi, evd)$ is finite, and from it we can construct a minimal logic proof of $\psi$.
\end{enumerate}
\begin{figure}[h]
\begin{mathpar}
\inferrule[$\wedge$pair]
{H\models G_1 \wedge G_2, \mypair{ evd_1}{evd_2}}
                {H\models G_1,evd_1 \\ H\models G_2,evd_2} \and
\inferrule[$\exists$pair]
                {H\models \myex{x}{G}, \mypair{ evd_1}{evd_2} }
                {H\models \myex{x}{G},\cbvpair{evd_1}{evd_2} } \and
\inferrule[$\exists$val pair]
        {d:\! D\in H\models \myex{x}{G}, \cbvpair{d}{evd}\ ({\it d\ a\ variable})}
                {H \models \subst{G}{x}{d},evd} \\
\inferrule[$\vee$inl]
        {H\models G_1 \vee G_2, \inl{evd}}
                {H\models G_1,evd} \and
\inferrule[$\vee$inr]
        {H\models G_1 \vee G_2, \inr{evd}}
                {H\models G_2,evd} \and
\inferrule[$\Rightarrow$$\lambda$]
        {H\models G_1 \Rightarrow G_2, \lam{x}{evd}}
                {H;x:\!G_1\models G_2,evd} \and
\inferrule[$\forall$$\lambda$]
        {H\models \myall{y}{G}, \lam{x}{evd}}
                {H;d:\!D\models \subst{G}{y}{d},evd}
\end{mathpar}
The bound variable in $\lam{x}{evd}$ in \TirName{$Rightarrow\lambda$} is renamed to avoid variables in $H$
and the variable $d$ in \TirName{$\forall\lambda$} is fresh.
\caption{
Rules for evidence structures with canonical evidence.}
  \label{fig:canonical}
\end{figure}

The first seven derivation rules shown in Figure~\ref{fig:canonical}
match evidence structures where the evidence is in canonical form.
\begin{mydef}
  An evidence derivation rule is valid if for any evidence structure matching the pattern above the line, the derived instances below the line
  are evidence structures.
\end{mydef}
\begin{mylemma}
  The rules in Figure~\ref{fig:canonical} are valid.
\end{mylemma}
\begin{proof}
  Since these rules do not add constraints to the context, we only have to prove that the derived evidence term is
  consistent with the derived goal.

  For the rule \TirName{$\wedge$pair}, suppose $M \models H$, then $\mypair{ evd_1}{evd_2}$ is consistent with $M(G_1 \wedge G_2)$.
  So,
  \begin{eqnarray*}
  i_{M(G_1) \times M(G_2)}(\mypair{ evd_1}{evd_2})\downarrow  &\Rightarrow&\\
  \cbvapp{\lambda a.\cbvapp{\lambda b.\mypair{a}{b}}{i_{M(G_2)}(evd_2)}}{i_{M(G_1)}(evd_1)}\downarrow  &\Rightarrow&\\
  i_{M(G_1)}(evd_1)\downarrow\  \wedge\ i_{M(G_2)}(evd_2)\downarrow
  \end{eqnarray*}

  For the rule \TirName{$\Rightarrow\lambda$}, suppose $M \models H;x:\!G_1$, then $\lam{ x}{evd}$ is consistent with $M(G_1 \Rightarrow G_2)$.
  So,
  \begin{eqnarray*}
  i_{M(G_1) \rightarrow M(G_2)}(\lam{x}{evd})\downarrow  &\Rightarrow&\\
  \graph{\lambda a.i_{M(G_2)}(evd(a))}\downarrow &\Rightarrow&\\
  \myall{a \in M(G_1)}{i_{M(G_2)}(evd(a))}\downarrow &\Rightarrow&\\
  i_{M(G_2)}(evd(M(x)))\downarrow
  \end{eqnarray*}
  The proofs of validity of the other rules for canonical evidence are similar to these.
\end{proof}
\begin{figure}[!h]
  \begin{mathpar}
\inferrule[var]
                {H_1;v:G;H_2\models G, v}
                {  } \and
    \inferrule[decide]
    {H_1;c:\!A\vee B;H_2\models G, \principal{\decide{\underline{c}}{x.a}{y.b}}}
                 {\subst{(H_1;x:\!A;H_2\models G, \newprincipal{a})}{c}{\inl{x}} \\
                  \subst{(H_1;y:\!B;H_2\models G, \newprincipal{b})}{c}{\inr{y}}} \and
    \inferrule[$\wedge$spread]
    {H_1;p:\!A\wedge B;H_2\models G, \principal{\spread{\underline{p}}{x,y.t}}}
    {\subst{(H_1;x:\!A;y:\!B;H_2\models G, \newprincipal{t})}{p}{\mypair{x}{y}} } \and
    \inferrule[$\exists$spread]
    {H_1;p:\myex{z}{P};H_2\models G, \principal{\spread{\underline{p}}{x,y.t}}}
    {\subst{(H_1;d:\!D;y:\!\subst{P}{z}{x};H_2\models G, \newprincipal{t})}{p}{\mypair{d}{y}} }  \and
    \inferrule[apply const]
    {f = \const{v} \in H \models G, \principal{\underline{f}(t)}}
                {H\models G, \newprincipal{v}} \and
    \inferrule[$\Rightarrow$apply]
    {\not\exists v. f = \const{v} \in H_1;f:\!A \Rightarrow B;H_2\models G, \principal{\underline{f}(t)}}
    {H_1;f:\!A \Rightarrow B;H_2\models A, t \\ H_1;f:\!A \Rightarrow B;H_2;v:\!B;f = \const{v}\models G, \newprincipal{v}} \and
    \inferrule[$\forall$apply]
    {f:\!\myall{z}{P} \in H\models G, \principal{\underline{f}(t)}}
                 {H\models G, \newprincipal{\cbvap{f}{t}} } \and
   \inferrule[apply model]
   {\cbvap{f}{d}=t \in H\models G, \principal{\cbvap{f}{\underline{d}}}}
                 {H\models G, \newprincipal{t} } \and
    \inferrule[$\forall$cbv]
    {\not\exists t.\, \cbvap{f}{d}=t \in H, \conts{f:\!\myall{z}{P},d:\!D}{H} \models G, \principal{\cbvap{f}{\underline{d}}}}
                 { H;w:\!\subst{P}{z}{d};\cbvap{f}{d}=w\models G, \newprincipal{w} }
  \end{mathpar}
  The bound variables, $d$,  $x$, and $y$, in rules \TirName{decide} and \TirName{spread} are renamed to avoid variables in $H$.
  The variables, $v$ and $w$ , introduced in rules \TirName{$\Rightarrow$apply} and \TirName{cbv new} are fresh.
  \caption{Rules for evidence structure with non canonical evidence.}
  \label{fig:noncanonical}
\end{figure}
The remaining rules match evidence that is not in canonical form. If a term is not in canonical form but some instance of it will compute to canonical form
then the term must have a subterm that is a variable and the computation depends on the value of that variable in order to proceed.
We call such a variable the {\em principal variable} and any subterm in such a position a {\em principal subterm}

\begin{mydef}
  The principal subterm $\prin{t}$ of term $t$ is defined inductively by:
  \begin{mathpar}
    \prin{\decide{d}{x.a}{y.b}} = \prin{d} \\
    \prin{\spread{p}{x,y.b}} = \prin{p} \\
    \prin{f(b)} = \prin{f} \\
    \prin{\cbvap{f}{b}} = \prin{b} \\
    \prin{\cbvpair{a}{b}} = \prin{a} \\
    \prin{t} = t, {\it otherwise}
  \end{mathpar}
  We write $\genprincipal{t}{\underline{x}}$ when $x$ is the principal subterm of $t(x)$.
\end{mydef}

The rules shown in Figure~\ref{fig:noncanonical} match on the
operator that is applied to the principal variable in the evidence.
When a fresh variable from $\Var'$ is needed, we take the least
index greater than all the variables already in use. This will
guarantee that all the constraints remain stratified.

\begin{mylemma}
  The constraints in the evidence structures derived from the rules in Figure~\ref{fig:noncanonical} are unique, stratified constraints.
  \label{lemma:ok}
\end{mylemma}

\begin{proof}
  Only the rules \TirName{$\Rightarrow$apply} and \TirName{$\forall$cbv}
  add  new constraints and they apply only when there is not already a similar constraint.
  The constraints are changed only by the rules (\TirName{decide}, \TirName{spread}, and \TirName{$\Rightarrow$apply}) that substitute
  a pattern ($\inl{x}$, $\inr{y}$, or $\mypair{x}{y}$) for a variable. In each case, the new variables introduced are fresh,
  and substituting a pattern for a variable in a pattern results in a pattern, so all the constraints remain unique, stratified patterns.
\end{proof}
\begin{mylemma}
  The rules in Figure~\ref{fig:noncanonical} are valid
  \label{lemma:nv}
\end{mylemma}
\begin{proof}
  Because it depends on the restriction to tight models,
  we consider first the proof of
  \begin{mathpar}
    \inferrule[$\Rightarrow$apply]
    {\not\exists v. f = \const{v} \in H_1;f:\!A \Rightarrow B;H_2\models G, \principal{\underline{f}(t)}}
    {H_1;f:\!A \Rightarrow B;H_2\models A, t \\ H_1;f:\!A \Rightarrow B;H_2;v:\!B;f = \const{v}\models G, \newprincipal{v}}
  \end{mathpar}
  Assume that the structure above the line is an evidence structure, and let $M \models_t H_1;f:\!A \Rightarrow B;H_2$.
  Then $M(\principal{\underline{f}(t)})$ is consistent with $M(G)$.
  Since $i_{M(G)}$ is strict, this implies that $M(t)$ is in the domain of $M(f)$ and
  since $M$ is tight, $M(t)$ is consistent with $M(A)$. This proves the validity of the first derived structure
                $H_1;f:\!A \Rightarrow B;H_2\models A, t$

  If $M \models_t H_1;f:\!A \Rightarrow B;H_2;v:\!B;f = \const{v}$ then the model $M$ is also a tight model of $H_1;f:\!A \Rightarrow B;H_2$ so
   $M(\principal{\underline{f}(t)})$ is consistent with $M(G)$ and this implies that
   $M(\principal{v})$ is consistent with $M(G)$  because $M(f(t))$ must converge to $M(v)$.
  This proves the validity of the second derived structure and finishes the proof of the rule \TirName{$\Rightarrow$apply}

   Consider next the rule
   \begin{mathpar}
    \inferrule[$\forall$cbv]
    {\not\exists t.\, \cbvap{f}{d}=t \in H, \conts{f:\!\myall{z}{P},d:\!D}{H} \models G, \principal{\cbvap{f}{\underline{d}}}}
                 { H;w:\!\subst{P}{z}{d};\cbvap{f}{d}=w\models G, \newprincipal{w} }
   \end{mathpar}
   If $M \models_t H;w:\!\subst{P}{z}{d};\cbvap{f}{d}=w$ then $M \models_t H$ and because $M(D)$ is a value type,
   $M(\cbvap{f}{d}) = M(f(d)) = M(w) \in M(P(d))$. Since $M(\principal{\cbvap{f}{\underline{d}}}) $ is consistent with $M(G)$ and
   $i_{M(G)}$ is strict, this implies that $M(\principal{w}) $ is consistent with $M(G)$.
   So \TirName{$\forall$cbv} is a valid rule.

   Consider next the rule
   \begin{mathpar}
    \inferrule[$\exists$spread]
    {H_1;p:\myex{z}{P};H_2\models G, \principal{\spread{\underline{p}}{x,y.t}}}
    {\subst{(H_1;d:\!D;y:\!\subst{P}{z}{x};H_2\models G, \newprincipal{t})}{p}{\mypair{d}{y}} }
   \end{mathpar}
   If $M \models_t {\subst{(H_1;d:\!D;y:\!\subst{P}{z}{x};H_2)}{p}{\mypair{d}{y}} }$ then the model
   \[M' = M\left[ p := \mypair{M(d)}{M(y)} \right]\] is
   a tight model of $H_1;p:\myex{z}{P};H_2$, so $M'(\principal{\spread{\underline{p}}{x,y.t}}) $ is consistent with $M'(G)$. This implies
   that $M(\subst{\newprincipal{t})}{p}{\mypair{d}{y}} )  $ is consistent with $M'(G) = M(G)$.

   The proofs for the validity of the remaining rules are similar to these.
\end{proof}
\begin{mylemma}
  If $H \models G, evd$ is an evidence structure, and $evd'$ is obtained by computing $evd$ until it is canonical or has a principal variable,
  then $H \models G, evd'$ is an evidence structure.
  \label{lemma:compute}
\end{mylemma}
\begin{proof}
  If $M \models_t H$ then $M(evd)$ is consistent with $M(G)$ so $(i_{M(G)}(evd)\downarrow)$. This implies
  $(i_{M(G)}(evd')\downarrow)$ since the computations are the same.
\end{proof}
\begin{mylemma}
  If $H \models G, evd$ is an evidence structure, and $evd$ is canonical or a principal variable, then
  $H \models G, evd$ matches one of the sixteen rules in Figure~\ref{fig:canonical} and Figure~\ref{fig:noncanonical}
  \label{lemma:match}
\end{mylemma}
\begin{proof}
  By Lemma~\ref{satisfiable} there is a tight model $M \models_t H$. Thus, $M(evd)$ is consistent with $M(G)$. If $evd$ is
  canonical, then $H \models G, evd$ must match one of the rules in Figure~\ref{fig:canonical} because the type $M(G)$ must be
  a product, union, or function type.

  If $evd$ has a principal variable $v$ then $v:\!T \in H$ for some $T\in \MFml$ and $M(v) \in M(T)$. Since $v$ is principal
  and $i_{M(G)}$ is strict, the computation $i_{M(G)}(M(evd))$ must reduce the subterm of $M(evd)$ containing $M(v)$.
  Since $M(T)$ must be a product, union, or function type, only a spread, decide, apply, or call-by-value apply redex can
  apply. Therefore one of the rules in Figure~\ref{fig:noncanonical} must match $H \models G, evd$.
\end{proof}
The preceding lemmas show that there is a well defined procedure that starts with the evidence
structure $( \models \psi, evd)$ constructed from uniform evidence for $\psi$ and recursively
builds a tree of evidence structures by alternating computation of $evd$ until it is canonical or has a principal variable with
matching the evidence structure against the sixteen derivation rules and applying the derivation.

It is routine to check that each derivation corresponds to a proof rule of minimal logic.
In our implementation of the proof procedure (shown in the Appendix) we
need only the evidence term $evd$ and the constraints (which we call the ``model'') because the
typing and the goal are just the current hypotheses and goal of the sequent being proved. From this information the
recursive Nuprl tactic decides which derivation rule to apply (or that it needs to compute the evidence term) and then
updates the evidence and constraints and uses one of the
primitive logical rules to get the next typing hypotheses and next goal term.

Our Theorem~\ref{thm:min} is proved once we establish that the
recursive procedure terminates.
\section{Termination of the Proof procedure}
We first show termination under the assumption that $evd$ is fully
normalized, which will establish Theorem~\ref{thm:normal}.
\begin{mylemma}
  If $evd$ is fully normalized then the evidence structure generation procedure terminates.
  \label{lemma:terminates}
\end{mylemma}
\begin{proof}
  Let $nc(evd)$ be the number of occurrences of {\bf decide}, {\bf spread}, or {\bf apply} operators in term $evd$.
  Let $cbv(evd)$ be the number of occurrences of the call-by-value apply operator.
  Let $npr(evd)$ be the number of occurrences of the $\mypair{x}{y}$ operator .
  Let $cn(evd)$ be the number of occurrences of the $\cbvpair{x}{y}$, $\inl{x}$, or $\inr{y}$ operators .

  We prove termination by induction on the lexicographically ordered tuple
  \[\langle nc(evd), cbv(evd), npr(evd), cn(evd) \rangle\]

  Each rule in Figure~\ref{fig:canonical} changes $evd$ to a subterm of $evd$ and removes at least one of the counted operators except for
  rule \TirName{$\exists$pair} which changes a $\mypair{x}{y}$ into a $\cbvpair{x}{y}$, so the measure decreases in each of these steps.

  Some of the rules in Figure~\ref{fig:noncanonical} reduce the measure by replacing a subterm of $evd$ that includes a
  {\bf decide}, {\bf spread}, or {\bf apply} operator by a variable and then substituting a pattern into the result.
  This reduces the $nc(evd)$ count and may increase only the $npr(evd)$ and $cn(evd)$ counts because patterns have only
  $\mypair{x}{y}$ , $\inl{x}$, or $\inr{y}$ operators .
  Thus, in every case it is easily checked that the measure decreases.

  It remains to show that in the computation steps that
  compute $evd$ until it is canonical or has a principal variable we can in fact fully normalize the evidence term and
  that this will not increase the measure.

  If $evd$ is fully normalized, then the only rules which derive evidence that may not be fully normalized are those that
  substitute a pattern. The resulting term $evd'$, which has some pattern $ptn$ in some places where $evd$ had a variable,
  can contain only spread, decide,  apply, or call-by-value apply redexes. When these are reduced, they result only in
  sub-patterns of $ptn$ being substituted for variables. Thus, by induction on the size of $ptn$ we can show that
  normalization of $evd'$ terminates and does not increase the measure.
\end{proof}

The proof of termination for the general case where $evd$ is not
assumed to be fully normalized uses Brouwer's Fan Theorem. For that
proof we need the following definitions and lemmas.
\begin{mydef}
  A derivation rule is {\em constant domain} if it does not add a new variable $d_i:\!D$ to the derived contexts.
\end{mydef}
All of the derivation rules in Figures~\ref{fig:canonical} and
~\ref{fig:noncanonical} are constant domain except for the rules
\TirName{$\forall\lambda$} and \TirName{$\exists$spread}.
\begin{mydef}
  A derivation is a $\psi$-deriviation if it is an instance of one of the derivation rules where
  the formulas in the context $H$ and goal $G$ are instances of subformulas of $\psi$.
\end{mydef}

\begin{mydef}
  Context $H'$ is a {\em constant domain $\psi$-extension} of context $H$ (written $H <_{cd}^{\psi} H'$)
  if $H'$ can be obtained from $H$ by applying $\psi$-derivations that are instances of
  constant domain derivation rules.

  Context $H$ is a {\em maximal $\psi$-context} if there is no proper $H'$ with $H <_{cd}^{\psi} H'$.
\end{mydef}
\begin{mylemma}
 For any formula $\psi \in \MFml$, and any context $H$ there are only finitely many $H'$ such that
 $H <_{cd}^{\psi} H'$
  \label{lemma:cdext}
\end{mylemma}
\begin{proof}
  Let $D_0$ be the set of variables $d_i:\!D \in H$.
 Repeated application of the constant domain $\psi$-derivations will add new declarations and constraints
 $v:\!P(d); \cbvap{f}{d} = v$ for all the universally quantified declarations $f:\!\myall{x}{P(x)}$ and every
 $d\in D_0$.  These new declarations will in turn be instantiated with every $d \in D_0$. Any declarations
 of the form $v:\! A \vee B$ will generate two derived extensions
 where $v$ is replaced by either $\inl x$ for $x:\!A$ or by $\inr y$ for $y:\!B$.
 Declarations of the form $p:\! A \wedge B$ will be replaced by $x:\!A; y:\!B$ and $p$ replaced by $\mypair{x}{y}$.
 Every subformula of $\psi$ may be added to the context with its free variables replaced by members of $D_0$.
 But for a finite $D_0$ and fixed formula $\psi$ there are only finitely many such extensions.
\end{proof}
\begin{mycor}
  For any context $H$ there is a finite, non-empty set of maximal $\psi$-contexts $H'$ such that
 $H <_{cd}^{\psi} H'$
\end{mycor}
\newcommand\SM{\mathcal{SM}(\psi)}
\begin{mydef}
  The {\em one step $\psi$-extension} of $H$ is obtained from $H$ by adding $d_i:\!D$ for
  the least $i\in\mathbb{N}$ for which $d_i$ is not in $H$, then applying the \TirName{$\exists$spread} rule to
  add new domain elements for every existentially quantified formula in $H$.

  Context $H'$ is a {\em next $\psi$-extension} of $H$ if it is a maximal constant domain $\psi$-extension
  of the one step $\psi$-extension of $H$.

  $\SM$, the {\em spread of symbolic models} of $\psi$ is the tree with the empty context at the root
  and the successors of node $H$ being the next $\psi$-extensions of $H$.
\end{mydef}
An infinite path through $\SM$ describes a freely-chosen model $M$
with $M(D) = \Var$. In this model the evidence term $evd$ must
compute evidence for $M(\psi)$ and we use the termination of this
computation to bar the spread $\SM$. Brouwer's Fan Theorem then
gives a uniform bar and this implies that our proof procedure
terminates on $evd$ and produces a minimal logic proof of $\psi$.
\begin{mydef}
  Let $\alpha$ be an infinite path in $\SM$.
  The computation $c(\alpha,\psi,evd, n)$ where $n > 0$, is defined by
  computing $evd$ in the context $\alpha(n)$ (a maximal context along path $alpha$)
  to a term $evd'$ that is canonical or has a principal variable. The computation proceeds by cases:
  \begin{itemize}
    \item if $evd'$ is a variable $v$ and $v:\!\psi$ is in the context then halt and return $n$.
    \item If $evd'$ is $\inl evd_1$ and $\psi = \psi_1 \vee \psi_2$ then the computation proceeds
  with $c(\alpha, \psi_1, evd_1,n)$.
    \item If $evd'$ is $\inr evd_2$ and $\psi = \psi_1 \vee \psi_2$ then the computation proceeds
  with $c(\alpha, \psi_2, evd_2,n)$.
\item If $evd'$ is $\mypair{evd_1}{evd_2}$ and $\psi = \psi_1 \wedge \psi_2$ then return the
  maximum of the dovetailed or sequential computation of both $c(\alpha, \psi_1, evd_1,n)$ and $c(\alpha, \psi_2, evd_2,n)$.
\item If $evd'$ is $\lambda x.\, evd_1$ and $\psi = \psi_1 \Rightarrow \psi_2$ then since the context $\alpha(n)$ is
  maximal there is a declaration $v:\! \psi_1$, so proceed with $c(\alpha, \psi_2, evd_1[x:= v],n)$.
\item If $evd'$ is $\lambda x.\, evd_1$ and $\psi = \myall{x}{\psi_2}$ then in $\alpha(n+1)$ a fresh $d_j:\!D$ was added
  so proceed with $c(\alpha, \psi_2[x:= d_j], evd_1[x:= d_j],n+1)$.
\item If $evd'$ has a principal variable $v$ that is the argument to a ${\bf decide}$ operator, then
  the maximal context specifies that $v= \inl{x}$ or $v = \inl{y}$, so replace $v$ and proceed with
  the computation.
\item If $evd'$ has a principal variable $v$ that is the argument to a ${\bf spread}$ operator, then
  if the maximal context specifies that $v = \mypair{x}{y}$ replace $v$ and proceed with
  the computation (this must happen if $v:\!A \wedge B$ in the context).
  Otherwise, if $v:\! \myex{z}{P(z)}$ in the context then in the next context, $\alpha(n+1)$,
  $v$ will be replaced by a pair $\mypair{d_j}{w}$ where $d_j:\!D$ is new and
  $w:\!P(d_j)$, so replace $v$ by $\mypair{d_j}{v}$ to get $evd''$ and proceed with $c(\alpha,\psi, evd'',n+1)$.
\item If $evd'$ has a principal variable $v$ where the principal subterm is $v(d_n)$ then $d_n:\! D$ is in the context,
  and since the context is maximal there is a constraint
  $v(d_n) = w$ in the context, so replace the subterm $v(d_n)$ with $w$ and proceed with the computation.
\item Otherwise abort the computation.
  \end{itemize}
\end{mydef}
\begin{mylemma}
  If $evd$ is uniform evidence for $\psi$ then the computation $c(\alpha,\psi,evd, n+1)$ converges.
  \label{lemma:conv}
\end{mylemma}
\begin{proof}
  Any path $\alpha$ through $\SM$ defines a model $M$ in which $evd \in M(\psi)$
\end{proof}
\begin{mycor}
  The proof procedure terminates and this establishesTheorem~\ref{thm:min}
\end{mycor}
\begin{proof}
  For any $\alpha$, the computation $n = c(\alpha,\psi,evd, 1)$ converges.
  This defines a bar on the fan $\SM$. By Brouwer's theorem, there is a uniform bar $N$.
  The length of any branch in the tree of evidence structures produced by the proof procedure is bounded by
  $c(\alpha,\psi,evd, 1)$ for some path $\alpha$. Thus the height of the tree of evidence structures is bounded by $N$.
  Since it is finitely branching, it is finite.
\end{proof}

We have implemented the proof procedure as a tactic in Nuprl and tested it on a number of examples.
We can construct evidence terms from the extracts of Nuprl proofs or construct them by hand. We can then
modify the evidence terms using any operators we like so that the resulting term is computationally equivalent
to the original. Thus we can introduce abbreviations (which is equivalent to using the \TirName{cut} rule) and
use operators such as $\pi_1$ and $\pi_2$ (which Nuprl displays as {\tt fst} and {\tt snd} as in ML) and
$\ite{c}{a}{b}$ that are defined in terms of the primitive ${\bf spread}$ and ${\bf decide}$ operators.
In appendix~\ref{ap1} we show one such example and describe the implementation of the tactic.

\section{Observations and Corollaries}
If $evd_1$ is uniform evidence for $\psi_1$ and $evd_2$ is uniform evidence for $\psi_1 \Rightarrow \psi$ then
the application $evd_2(evd_1)$ is uniform evidence for $\psi$. This observation gives us a semantic proof of cut elimination for
first order minimal logic.
\begin{mylemma}
  If $\psi \in \MFml$ is provable in minimal logic with the cut rule ($\vdash_{MLC} \psi$) then $\vdash_{ML} \psi$
  \label{cut elim}
\end{mylemma}
\begin{proof}
  The evidence term extracted from the proof $\vdash_{MLC} \psi$ is uniform evidence for $\psi$. ByTheorem~\ref{thm:min},
  $\vdash_{ML} \psi$
\end{proof}

\newpage
\section{Appendix}
\label{ap1}

\begin{figure}[h]
  \begin{center}
  \begin{minipage}[b]{\linewidth}
    \footnotesize

\begin{program*}
\>\\
\>\mvdash{} \mforall{}[A,D:Type]. \mforall{}[R,Eq:D {}\mrightarrow{} D {}\mrightarrow{} \mBbbP{}].\\
\>    ((\mforall{}x,y,z:D.  (R[x;y] {}\mRightarrow{} (R[y;z] \mvee{} Eq[y;z]) {}\mRightarrow{} R[x;z]))\\
\>    {}\mRightarrow{} (\mforall{}x:D. (R[x;x] {}\mRightarrow{} A))\\
\>    {}\mRightarrow{} (\mforall{}x:D. \mexists{}y:D. R[x;y])\\
\>    {}\mRightarrow{} (\mexists{}m:D. \mforall{}x:D. ((Eq[x;m] {}\mRightarrow{} A) {}\mRightarrow{} R[x;m]))\\
\>    {}\mRightarrow{} A)\\
\> \\
\>BY EvidenceTac \mkleeneopen{}\mlambda{}Trans,Irr,Unbdd,MxEx.\\
\>                  let m = fst(MxEx) in\\
\>                   let bounds = snd(MxEx) in\\
\>                   let y,ygtr = Unbdd m \\
\>                   in let loop = Trans m y m ygtr in\\
\>                       let F = \mlambda{}x.(Irr m (loop x)) in\\
\>                       F (inl (bounds y (\mlambda{}eq.(F (inr eq )))) )\mkleeneclose{}\mcdot{}\\
\>   THENA Auto\\
\>
\end{program*}
  \end{minipage}
  \end{center}
  \caption{Example minimal logic proof from evidence}
  \label{fig:exm}
\end{figure}
This example shows how equality can be represented as an atomic relation symbol. The formula states (in minimal logic) that
an irreflexive, transtitive relation that is unbounded cannot have a maximal element. We have introduced a number of abbreviations
into the evidence term to illustrate the fact that the proof procedure does not require normalized terms.

The tactic {\tt EvidenceTac}  is shown in Figure~\ref{fig:tactic}.
It uses the evidence to generate the proof. In Nuprl, some of the
primitive rules of minimal logic (hypothesis, and, or, implies,
forall, exists introduction and elimination) create auxilliary
subgoals to show that the rules have been applied to well-formed
propositions. In the proof in Figure~\ref{fig:exm} the tactic {\tt
THENA Auto} is used to prove these auxilliary goals.

\begin{figure}[!ht]
  \begin{center}
  \begin{minipage}[b]{\linewidth}
    \footnotesize
\begin{program*}
\>let EvidenceTac evd  =\\
\>   -- helper functions here -- \\
\>letrec evdProofTac M evd  p =\\
\>  let op = opid\_of\_term evd in\\
\>  if member op ``variable pair inl inr lambda`` then\\
\>      canonical op M evd p\\
\>  else\\
\>  let t = get\_principal\_arg\_with\_context evd in\\
\>  if is\_variable\_term (subtermn 1 t) then\\
\>   let op = opid\_of\_term t in\\
\>   if member op ``spread decide callbyvalue apply`` then\\
\>       noncanonical op t M evd p\\
\>   else (AddDebugLabel `arg not reducible` p)     \\
\>  else let evd' = apply\_conv (ComputeToC []) evd in\\
\>        if alpha\_equal\_terms evd' evd then Id p\\
\>        else evdProofTac M evd' p\\
\>      \\
\>in Repeat UniformCD\\
\>   THEN evdProofTac [] evd \\
\>;;\\
\>
\end{program*}
  \end{minipage}
  \end{center}
  \caption{Tactic code for Proof from Uniform Evidence}
  \label{fig:tactic}
\end{figure}
The basic structure of the tactic is to take off the uniform quantifiers and then start the proof procedure from evidence.
If the evidence is canonical it uses one of the rules for that case, otherwise if there is a principal variable it uses one of
the rules for non-canonical evidence, and otherwise it computes the evidence term.
\begin{figure}[!ht]
  \begin{center}
  \begin{minipage}[b]{\linewidth}
    \footnotesize
  \begin{program*}
\>let UniformCD p  = if is\_term `uall` (concl p) \\
\>                   then  (D 0 THENA Auto) p \\
\>                   else  Fail p in\\
\>let mk\_cbv\_pair t1 t2 = \\
\>    subst [`x',t1;`y',t2] \mkleeneopen{}let a := x in\\
\>                           <a, y>\mkleeneclose{} in\\
\>let mk\_cbv\_ap fun arg = \\
\>    subst [`arg',arg;`f',fun] \mkleeneopen{}let a := arg in\\
\>                               f a\mkleeneclose{} in\\
\>let do\_update v pattern redex result evd M = \\
\>     let sub = [v, pattern] in\\
\>     subst sub (replace\_subterm redex result evd),\\
\>     map (\mbackslash{}(ap,val). (ap, subst sub val)) M in\\
\>let lookup M t = \\
\>   let test (ap, val) = \\
\>     if alpha\_equal\_terms ap t then val else fail in\\
\>   inl (first\_value test M) ? inr () in\\
  \end{program*}
  \end{minipage}
  \end{center}
  \caption{Code for helper functions}
  \label{fig:helper}
\end{figure}
The helper code include the tactic for taking off a uniform quantifier, functions for forming the call-by-value pair and apply terms,
and code for substituting a pattern into the evidence and constraints (here called the model) in order to eliminate a redex from the
non-canonical evidence. The lookup function checks for the existence of a constraint on a given apply term from the evidence.

\begin{figure}[!ht]
  \begin{center}
  \begin{minipage}[b]{\linewidth}
    \footnotesize
  \begin{program*}
\>canonical op M evd p =\\
\> if op = `variable` then \\
\>      let x = dest\_variable evd in \\
\>        let n = get\_number\_of\_declaration p x in NthHyp n p \\
\> else if op = `pair` then\\
\>    let evd1,evd2 = dest\_pair evd in\\
\>    if is\_term `and` (concl p)  then\\
\>     (D 0 THENL [evdProofTac M evd1; evdProofTac M evd2]) p\\
\>    else if is\_term `variable` evd1 then\\
\>         (With evd1 (D 0) THENM (evdProofTac M evd2)) p\\
\>    else (evdProofTac M (mk\_cbv\_pair evd1 evd2)) p\\
\> else if op = `inl` then \\
\>      (OrLeft THENM (evdProofTac M (dest\_inl evd))) p\\
\> else if op = `inr` then \\
\>      (OrRight THENM (evdProofTac M (dest\_inr evd))) p \\
\> else let x,t = dest\_lambda evd in\\
\>      let z = maybe\_new\_var x (declared\_vars p) in\\
\>      let evd1 = if z = x then t else subst [x, mvt z] t in\\
\>      SeqOnM [D 0 ; RenameVar z (-1); evdProofTac M evd1] p\\
  \end{program*}
  \end{minipage}
  \end{center}
  \caption{Code for canonical evidence}
  \label{fig:ccase}
\end{figure}
The code for the canonical case comes from the rules in
Figure~\ref{fig:canonical}. In each case, the corresponding proof
rule of the logic is invoked with the tactic {\tt D 0}. To make life
easier for the users, Nuprl has organized all the primitive rules
into one tactic named {\tt D} (for decompose). The number {\tt 0}
indicates that we are applying a primitive rule to decompose the
conclusion of the sequent rather than one of the hypotheses. This is
because the canonical evidence always indicates that the next proof
step is an introduction rule.

\begin{figure}[!ht]
  \begin{center}
  \begin{minipage}[b]{\linewidth}
    \footnotesize
  \begin{program*}
\>noncanonical op t M evd p =\\
\> if op = `spread`  then\\
\>     let t1,bt = dest\_spread t in\\
\>     let v = dest\_variable t1 in\\
\>     let [x;y],body = rename\_bvs p bt in\\
\>     let pattern = mk\_pair\_term (mvt x) (mvt y) in \\
\>     let evd1, M' = do\_update v pattern t body evd M in\\
\>     let n = get\_number\_of\_declaration p v in \\
\>       Seq [D n\\
\>           ; RenameVar x n\\
\>           ; RenameVar y (n+1)\\
\>           ; evdProofTac M' evd1] p\\
\>  else if op = `decide`  then\\
\>     let t1, bt1, bt2 = dest\_decide t in\\
\>     let v = dest\_variable t1 in\\
\>     let [x],case1 = rename\_bvs p bt1 in\\
\>     let pattern1 = mk\_inl\_term (mvt x) in\\
\>     let evd1, M1 = do\_update v pattern1 t case1 evd M in\\
\>     let [y],case2 = rename\_bvs p bt2 in\\
\>     let pattern2 = mk\_inr\_term (mvt y) in\\
\>     let evd2, M2 = do\_update v pattern2 t case2 evd M in \\
\>     let n = get\_number\_of\_declaration p v in \\
\>      (D n THENL [ RenameVar x n THEN  evdProofTac M1 evd1\\
\>                 ; RenameVar y n THEN  evdProofTac M2 evd2\\
\>                 ]) p\\
\>  else if op = `callbyvalue`  then\\
\>     let kind, arg, ([x], B) = dest\_callbyvalue t in\\
\>     let B' = subst [x, arg] B in\\
\>     evdProofTac M (replace\_subterm t B' evd ) p\\
\>  else  apply\_case t M evd p\\
  \end{program*}
  \end{minipage}
  \end{center}
  \caption{Code for non-canonical evidence}
  \label{fig:nccase}
\end{figure}
The code for the non-canonical case comes from the rules in
Figure~\ref{fig:noncanonical}. In these cases we use an elimination
rule, indicated by the fact that the tactic calls on {\tt D n} where
{\tt n} is the hypothesis number for the declaration of the
principal variable. The code for the apply case is shown in
Figure~\ref{fig:applycase}. When the type of the declared variable
({\tt T = h n p}) is an implies we use the rule
\TirName{$\Rightarrow$apply} that adds a constraint that the
declared function is a constant function. In this implementation we
substitute the constant function for the variable and eliminate it
entirely. We can prove that this results in behavior that is
equivalent to the derivation rules.
\begin{figure}[h!]
  \begin{center}
  \begin{minipage}[b]{\linewidth}
    \footnotesize
\begin{program*}
\>apply\_case t M evd p = \\
\>  let fun,arg = dest\_apply t in \\
\>  let v = dest\_variable fun in\\
\>  let n = get\_number\_of\_declaration p v in\\
\>  let T = h n p in\\
\>  if is\_term `implies` T  then\\
\>    let x = maybe\_new\_var `x' (declared\_vars p) in\\
\>    let pattern = mk\_lambda\_term `z' (mvt x) in \\
\>    let evd1, M' = do\_update v pattern t (mvt x) evd M in\\
\>    ((D n  THEN Fold `implies` n)\\
\>       THENL [ evdProofTac M arg\\
\>             ; RenameVar x (-1) THEN  evdProofTac M' evd1]) p \\
\>  else if is\_term `all` T then\\
\>    if is\_variable\_term arg then\\
\>      let w = dest\_variable arg in\\
\>      let ans = lookup M t in\\
\>      if isl ans then\\
\>        evdProofTac M (replace\_subterm t (outl ans) evd) p\\
\>      else\\
\>        let x = maybe\_new\_var `x' (declared\_vars p) in\\
\>        let evd1 = replace\_subterm t (mvt x) evd in\\
\>        let M' = (t , (mvt x)).M in\\
\>        (SimpleInstHyp arg n THENM \\
\>          (Seq [ RenameVar x (-1); evdProofTac M' evd1])) p\\
\>   else let evd' = replace\_subterm t (mk\_cbv\_ap fun arg) evd in\\
\>        evdProofTac M evd' p\\
\> else (AddDebugLabel `fun in apply has wrong type` p)\\
\end{program*}
  \end{minipage}
  \end{center}
  \caption{Code for apply case of non-canonical evidence}
  \label{fig:applycase}
\end{figure}

\newpage
\bibliographystyle{plain}
\bibliography{rc-fk-2}

\end{document}